\documentclass[a4paper,USenglish,cleveref]{lipics-v2021}

\hideLIPIcs
\nolinenumbers

\usepackage{caption}
\usepackage{subcaption}

\usepackage[ruled]{algorithm}
\newcommand{\psfrage}[1]{{\color{blue}{\sf[PS: #1]}}}
\newcommand{\hpfrage}[1]{{\color{violet}\sf[HP: #1]}}
\newcommand{\fkfrage}[1]{{\color{teal}\sf[FK: #1]}}
\renewcommand{\psfrage}[1]{} \renewcommand{\hpfrage}[1]{} \renewcommand{\fkfrage}[1]{}

\usepackage{xcolor}
\usepackage{xspace}
\usepackage{booktabs}
\usepackage{dsfont}
\usepackage{footmisc}
\usepackage{marvosym}
\usepackage{amsmath}
\usepackage{hyperref}
\usepackage[capitalise,noabbrev]{cleveref}
\usepackage{tabularx}
\usepackage{listings}
\usepackage{multirow}

\usepackage{makecell}
\usepackage{placeins}
\usepackage{dblfloatfix}
\usepackage{midfloat}
\usepackage{afterpage}

\usepackage{pgfplots}
\pgfplotsset{compat=newest}

\usepgfplotslibrary{groupplots}
\pgfplotsset{every axis/.style={scale only axis}}

\pgfplotsset{
  major grid style={thin,dotted},
  minor grid style={thin,dotted},
  ymajorgrids,
  yminorgrids,
  every axis/.append style={
    line width=0.7pt,
    tick style={
      line cap=round,
      thin,
      major tick length=4pt,
      minor tick length=2pt,
    },
  },
  legend cell align=left,
  legend style={
    line width=0.7pt,
    /tikz/every even column/.append style={column sep=3mm,black},
    /tikz/every odd column/.append style={black},
  },
  legend style={font=\small},
  title style={yshift=-2pt},
  enlarge x limits=0.04,
  every tick label/.append style={font=\footnotesize},
  every axis label/.append style={font=\small},
  every axis y label/.append style={yshift=-1ex},
  /pgf/number format/1000 sep={},
  axis lines*=left,
  xlabel near ticks,
  ylabel near ticks,
  axis lines*=left,
  label style={font=\footnotesize},       
  tick label style={font=\footnotesize},
  plotBigComparison/.style={
    width=50.0mm,
    height=50.0mm,
  },
  plotPaCHashComparison/.style={
    width=50.0mm,
    height=40.0mm,
  },
  plotIoVolume/.style={
    width=50.0mm,
    height=40.0mm,
  },
  plotMaximumLoadFactor/.style={
    width=50.0mm,
    height=35.0mm,
  },
  plotSmallHist/.style={
    major grid style={white},
    minor grid style={white},
    ymajorticks=false,
    width=40.0mm,
    height=7.0mm,
  },
}

\usepackage{pgfplots}
\crefname{listing}{Algorithm}{Algorithms}
\crefname{lstlisting}{Algorithm}{Algorithms}
\Crefname{lstlisting}{Algorithm}{Algorithms}
\crefname{@theorem}{Theorem}{Theorems}
\Crefname{@theorem}{Theorem}{Theorems}

\newcommand{\myparagraph}[1]{\subparagraph*{#1}}
\newcommand{\anonymous}[1]{\ifx\authoranonymous\relax\textcolor{red}{Anonymous}\else{#1}\fi}

\newcommand{\punkt}{\text{.}}  %

\newcommand{\Is}{\ensuremath{\mathbin{:=}}}
\newcommand{\seq}[1]{\langle#1\rangle}
\newcommand{\binomial}[2]{\binom{#1}{#2}}
\newcommand{\Om}[1]{\mathrm{\Omega}\!\left(#1\right)}
\newcommand{\Oh}[1]{\mathrm{O}\!\left( #1\right)}
\newcommand{\Ohtilde}[1]{\ensuremath \tilde{\mathrm{O}}\!\left( #1\right)}

\newcommand{\Objects}{S}
\newcommand{\Bs}{\mathtt{B}}
\newcommand{\BsMinus}{\bar{\mathtt{B}}}
\newcommand{\Ms}{\mathtt{M}}
\newcommand{\scan}{\textit{scan}}
\newcommand{\sort}{\textit{sort}}
\newcommand{\ceil}[1]{\ensuremath{\lceil #1\rceil}}

\let\oldcite\cite
\renewcommand\cite{\unskip~\oldcite}

\newcommand{\mytitle}{PaCHash: Packed and Compressed Hash Tables}
\title{\mytitle}

\author{Florian Kurpicz}{Karlsruhe Institute of Technology, Germany}{kurpicz@kit.edu}{https://orcid.org/0000-0002-2379-9455}{}
\author{Hans-Peter Lehmann}{Karlsruhe Institute of Technology, Germany}{hans-peter.lehmann@kit.edu}{https://orcid.org/0000-0002-0474-1805}{}
\author{Peter Sanders}{Karlsruhe Institute of Technology, Germany}{sanders@kit.edu}{https://orcid.org/0000-0003-3330-9349}{}

\newcommand{\myauthorrunning}{Kurpicz, Lehmann, Sanders}
\authorrunning{\myauthorrunning}
\Copyright{\myauthorrunning}
\hypersetup{
  colorlinks=true,
  pdftitle={\mytitle},
  pdfauthor={\myauthorrunning},
  pdfsubject={}
}


\supplement{All implementations presented in this paper and scripts to reproduce our experimental evaluation are available under the GPLv3 license.}
\supplementdetails[]{PaCHash and Separator implementation}{https://github.com/ByteHamster/PaCHash}
\supplementdetails[]{Scripts for reproduction of results}{https://github.com/ByteHamster/PaCHash-Experiments}

\begin{document}
\maketitle

\begin{abstract}
  We introduce PaCHash, a hash table that stores its objects contiguously in an array without intervening space, even if the objects have variable size.
  In particular, each object can be compressed using standard compression techniques.
  A small search data structure allows locating the objects in constant expected time.
  PaCHash is most naturally described as a static external hash table where it needs a constant number of bits of internal memory per block of external memory.
  Here, in some sense, PaCHash beats a lower bound on the space consumption of $k$-perfect hashing.
  An implementation for fast SSDs needs about 5 bits of internal memory per block of external memory, requires only one disk access (of variable length) per search operation, and has small internal search overhead compared to the disk access cost.
  Our experiments show that it has lower space consumption than all previous approaches even when considering objects of identical size.
\end{abstract}

\section{Introduction}

\emph{Hash tables} support constant time key-based retrieval of
objects and are one of the most widely used data
structures. \emph{Compressed} data structures store data in a space
efficient way, preferably approaching the information theoretical
limit, and support various kinds of operations without the need to
decompress the entire data structure first \cite{gog2014theory,agarwal2015succinct,ferragina2005indexing,zhang2018efficient}. There has
been intensive previous work on both subjects but, surprisingly, the
intersection leaves big gaps. There is a lot of work on hash tables
which need little more space than
just the stored objects themselves \cite{koppl2022fast,bender2021all,arbitman2010backyard,fotakis2005space,Knu73,CuckooHashing}. 
However, all these approaches are only space efficient
for objects of identical size which makes it impossible to compress
the objects with variable bit-length codes.  Currently, most hash tables
for objects of variable size store references from table entries to the
data which entails a space overhead of at least $\log N$ bits per
object, where $N$ is the total size of all objects in the table.
Throughout this paper, $\log x$ stands for $\log_2 x$.
See \cref{s:prelim} for an introduction of basic techniques and \cref{tab:symbols} for a summary of the notation.

PaCHash eliminates fragmentation by \emph{packing} the objects
contiguously in memory without leaving free space. This makes it
impossible to use the approach of most previous hash tables to
directly use the hash function value to (approximately) locate the
objects. Instead, PaCHash uses a highly space efficient search data
structure that translates hash function values to memory
locations. More precisely, objects are first hashed to \emph{bins}.
The bins are stored contiguously in $m$ \emph{blocks} of size $\Bs$.
PaCHash essentially stores one bin index per block using a searchable compressed representation which enables finding the block(s) where a bin is stored.
In \cref{s:structure}, we describe the data structure in
more detail and in \cref{s:analysis} we analyze it.
Basically, for a tuning parameter $a$, the expected number of block reads to retrieve an object $x$ of size $|x|$ is about $1+1/a+|x|/\Bs$
while the internal memory data structure needs $2+\log(a)$ bits per \emph{block}.
We also discuss even smaller representations.

Even though hash tables like PaCHash have applications in object stores,
there is little previous work on space efficient hash tables for objects of variable
size (see \cref{s:related}). For objects of identical size $s$, the most space efficient
previous solutions are based on minimal perfect hashing (MPH)
\cite{esposito2020recsplit,BelazzouguiBD09} and require a constant number of bits per
object. PaCHash approximates this when choosing $\Bs=s$, also
needing a (slightly larger) constant number of bits per object but
lower construction time.  The picture changes when we look at larger
block sizes $\Bs=ks$ and the corresponding approach of \emph{minimal
  $k$-perfect hashing (M$k$PH)} \cite{BelazzouguiBD09}. Now, PaCHash still needs
only a constant number of bits per block, while there is a \emph{lower bound} of $\Om{\log k}$ bits
per block using M$k$PH (see \cref{s:analysis}).

Another fundamental data structure related to variable size objects and PaCHash is the variable-bit-length array (VLA).
A VLA is an array that allows direct access to objects of variable size.
Oftentimes, VLAs are used to efficiently access variable-length codes, e.g., Elias-\(\gamma\) and -\(\delta\) codes~\cite{Elias74} or Golomb codes~\cite{Golomb1966Encodings}, see \cref{s:related}.

\Cref{s:variants} describes different implementation variants of PaCHash including fully internal and fully external versions as well as a variant that is usable as VLA.
\Cref{s:experiments} describes experiments for an external implementation.
\Cref{s:conclusion} summarizes the results and discusses possible
directions for further research.

\begin{table}[t]
    \caption{Symbols used in this paper}
    \label{tab:symbols}
      \begin{tabularx}{\columnwidth}{ l X }
        \toprule
        $\Objects$& Set of objects \\
        $n$       & Number of objects \\
        $N$       & Total size of objects (bits) \\
        $p$       & Internal index data structure\\
        $a$       & Tuning parameter: Bins per block\\
        $m=N/\BsMinus$ & Number of blocks\\
        $\Bs$     & Block size (bits)\\
        $\BsMinus=\Bs-d$& Payload data per block\\
        $d\in0..\log\Bs$ & Encoding-dependent number of bits to store position of first bin of block\\
        \bottomrule
      \end{tabularx}
\end{table}

\myparagraph{Our Contribution.}
In this paper, we design the new hash table PaCHash.
The data structure supports objects of variable size with space overhead close to competitors that only support objects of identical size.
We analyze it thoroughly in a variant of the external memory model.
Finally, we compare our implementation with competitors from the literature.
As close contenders, we also implement \emph{Separator~Hashing} \cite{gonnet1988external, larson1984file} and \emph{Cuckoo~Hashing} \cite{azar1994balanced,pagh2003basic} with adaptions that partially allow variable size objects.

\section{Preliminaries}\label{s:prelim}
\myparagraph{Monotonic Sequences and Bit Vectors.}\label{par:eliasFano}
The index data structure of PaCHash mainly consists of a compressed representation of a monotonically increasing sequence $p=\seq{p_1,\ldots,p_k}$ of integers in the range $1..U$.
Searching boils down to predecessor queries in $p$, i.e., given a query integer $i$, the largest sequence element $\leq i$ is returned.

A well-known practical solution is \emph{Elias-Fano coding} \cite{Elias74,Fano71} which splits each $p_i$.
The $\log(U/k)$ least significant bits are directly stored in an array $L$ requiring $k\log(U/k)$ bits of space.
The $\log(k)$ most significant bits form a monotonic sequence of integers $H=\seq{u_1,\ldots,u_k}$ in the range $0..k$.
$H$ is stored in a bit vector of size $2k+1$ where $u_i$ is represented as a 1-bit in position $i+u_i$.
The total space usage therefore is $k(2 + \log(U/k)) + 1$ bits.
A predecessor query in $p$ executes a $\textit{select}_0$ query in
$H$ (finding the $i$-th 0-bit in $H$) which locates a cluster of
entries in $L$ that must contain the sought element.  Using additional
space $o(k)$, $\textit{select}_0$ queries can be answered in
constant time \cite{clark1997compact}.
In contrast to the general case, we will show that searching the cluster takes expected constant time in our application.

One can also interpret $p$ as the positions of 1-bits in a sparse bit vector which
enables even more compact representations.  For example, using
Succincter \cite{Patrascu2008Succincter}, about $k(1.44 + \log(U/k)) + 1$ bits
are achievable which is almost information theoretically optimal. In \cref{s:entropyCoding},
we give an even more compact format exploiting additional structure in
the bit vector.

\myparagraph{Model of Computation.}
We describe our results in a variant of the external memory model \cite{VitShr94} adapted to a situation
where objects are compressed to variable length sequences of bits.
We have a \emph{fast memory} of size $\Ms$ \emph{bits}. Accesses to a large \emph{external memory} are \emph{I/Os}
to blocks of $\Bs$ consecutive bits.
In contrast to the original model, we analyze both I/Os and internal work.
$\scan(N)$ denotes the cost (I/Os \emph{and} internal work) of scanning $N$ bits of data.%
\footnote{The internal work may depend on the encoding of the data. For example, we may
need $\Theta(N)$ machine instructions, or, a faster encoding may enable bit-parallel processing in $\Oh{N/\log n}$.}
$\sort(N)$ denotes the cost of sorting $N$ bits.%
\footnote{This entails $(N/\Bs)(1+\lceil\log_{\Ms/\Bs}(N/\Ms)\rceil)$ I/Os.
In this paper algorithms with linear internal work are possible exploiting random integer keys.
The cost also includes (de)coding overhead as in $\scan$ operations.}
In particular, we are interested in a high load factor, which is $N$ divided by the total external space usage.

\section{Related Work}\label{s:related}
The following section introduces related data structures from the literature.
\Cref{tab:competitors} provides an overview over the most important parameters.
\label{s:lsmTree}
There are close contenders in the form of \emph{object stores} from the database literature.
BerkeleyDB \cite{olson1999berkeley} uses a B$^+$-Tree \cite{comer1979ubiquitous} of order $d$, where each node branches between $d$ and $2d$ times.
LevelDB \cite{google2021leveldb} and RocksDB \cite{facebook2021rocksdb} use a Log-Structured Merge tree \cite{o1996log}, which stores multiple levels of a static data structure with increasing size.
Insertions go into the first level and when a level gets too full, it is merged into the next level.
SILT's \emph{LogStore} \cite{lim2011silt}, Facebook \emph{Haystack} \cite{beaver2010finding} and \emph{FAWN} \cite{andersen2009fawn} simply store a pointer of size $\Omega(\log N)$ to each object.
Real world instances often store very small objects \cite{nishtala2013scaling}, so the pointers add a considerable amount of overhead.

\myparagraph{Sorted Objects.}
\emph{LevelDB}'s static part \cite{google2021leveldb} stores objects in key order, enabling range searches and common-prefix-compression.
\emph{SortedStore} in SILT \cite{lim2011silt} sorts the objects by their hashed key and uses entropy coded tries as an index.
Pagh \cite{pagh2003basic} proposes to sort the $n$ objects by a hash function with range~$\geq n^3$.
The internal memory stores the first hash function value mapped to each block.
This data structure can be queried using a predecessor data structure in time $\Oh{\log\log n}$.
A novel idea in PaCHash is that it uses a hash function range based on the total space $N$ instead of the number of objects $n$, which enables efficient queries and compact representation.

\begin{table}[t]
  \caption{Space efficient object stores from the literature. 
        To unify the notation, we convert all values so that they refer to objects of size $s=256$ bytes stored in blocks of $\Bs=4096$ bytes.
        Each block contains $\Bs/s=16$ objects.
        Top: Stores for objects of identical size.
        Can be used for objects of variable size by using indirection or for some methods by accepting significantly lower load factors.
        Bottom: Dedicated variable size object stores.
        This table also contains VLAs, even though those are a slightly different field.
        }
    \label{tab:competitors}
    \begin{tabularx}{\columnwidth}{c l r r l }
        \toprule
        & Method                                              & Internal memory & Load Factor & I/Os \\ \midrule
        \multirow{11}{*}{{\rotatebox[origin=c]{90}{fixed size}}} 
        & Extendible Hashing \cite{fagin1979extendible}       &  \(\log m\) bits/block & 90\%       & 1    \\ %
        & Larson et al. \cite{larson1985external}             &    96 bits/block & $<$96\%    & 1    \\ %
        & SILT SortedStore \cite{lim2011silt}                 &    51 bits/block & 100\%      & 1    \\ %
        & Linear Separator \cite{larson1988linear}            &     8 bits/block & 85\%       & 1    \\ %
        & Separator \cite{gonnet1988external, larson1984file} &     6 bits/block & 98\%       & 1    \\
        & Robin Hood \cite{celia1988external}                 &     3 bits/block & 99\%       & 1.3  \\ %
        & Ramakrishna et al. \cite{ramakrishna1989dynamic}    &     4 bits/block & 80\%       & 1    \\ %
        & Jensen, Pagh \cite{jensen2008optimality}            &     0 bits/block & 80\%       & 1.25 \\ %
        & Cuckoo \cite{azar1994balanced,pagh2003basic}        &     0 bits/block & $<$100\%   & 2    \\
        & \textbf{PaCHash}, $a=1$                             &     2 bits/block & 100\%      & 2\footref{fn:io} \\
        & \textbf{PaCHash}, $a=8$                             &     5 bits/block & 100\%      & 1.13\footref{fn:io} \\
        \midrule
        \multirow{6}{*}{{\rotatebox[origin=c]{90}{variable size}}} 
        & SILT LogStore \cite{lim2011silt}                    &   832 bits/block & 100\%      & 1    \\ %
        & K{\"{u}}lekci \cite{Kulekci2014VLA} (VLA)           &   176 bits/block & \(<\)100\% & 0--11\footref{fn:kulecki}\\ %
        & SkimpyStash \cite{debnath2011skimpystash}           &    32 bits/block & $\leq$98\% & 8    \\ %
        & Blandford, Blelloch \cite{blandford2008compact} (VLA) &  16 bits/block & \(\leq\)50\% & 1  \\ %
        & \textbf{PaCHash}, $a=1$                             &     2 bits/block & 99.95\%    & 2.06\footref{fn:io} \\ %
        & \textbf{PaCHash}, $a=8$                             &     5 bits/block & 99.95\%    & 1.19\footref{fn:io} \\
        \bottomrule
    \end{tabularx}
\end{table}

\myparagraph{External Hash Tables.}
In external hash tables, each table cell corresponds to a fixed size block.
A common technique to support variable size objects is using indirection by internally storing a pointer to the object contents, possibly inlining parts of the objects \cite[Section~4]{lim2011silt}.
NVMKV \cite{marmol2015nvmkv} and KallaxDB \cite{chen2021kallaxdb} use an SSD as one large hash table and rely on SSD internals to handle empty blocks in a space efficient way.
Overflowing blocks due to hash collisions can be handled with perfect hashing \cite{larson1985external, ramakrishna1989dynamic} or using one of the following techniques.

\addtocounter{footnote}{1}
\footnotetext{\label{fn:io}PaCHash performs one I/O of variable size which is faster than the competitors' multiple I/Os.}
\addtocounter{footnote}{1}
\footnotetext{\label{fn:kulecki}Using $256$ byte objects, we have an alphabet size of $2^{8 \cdot 256}$, and $\log\log2^{8 \cdot 256}=11$.}

With \emph{Hashing with Chaining}, objects of overflowing blocks are stored in linked lists.
\emph{SkimpyStash} \cite{debnath2011skimpystash} chains objects using an external successor pointer for each object.
This trades internal memory space for latency because of multiple dependent I/Os.
Jensen and Pagh's \cite{jensen2008optimality} data structure reserves parts of the external memory as a buffer to reduce the need for chaining.
\emph{Extendible Hashing} \cite{fagin1979extendible} keeps a balanced tree of blocks.
Overflowing blocks are split into two children indexing one more bit of the hashed key.

Another method for resolving collisions is \emph{open addressing}, where each object could be located in multiple blocks.
\emph{Cuckoo~Hashing} \cite{CuckooHashing,DieWei07} locates each object in one of two (or more \cite{fotakis2005space}) independently hashed blocks.
Queries can load both blocks in parallel to reduce latency.
With \emph{Separator~Hashing} \cite{gonnet1988external, larson1984file}, each object has a sequence of blocks it could be stored in and a corresponding sequence of signatures.
When a block overflows, the objects with the highest signature values are pushed out to the next block in their respective sequences.
The internal memory stores the highest signature value of the objects placed in each block.
A query follows the object's sequence of blocks and stops when it finds a separator that is larger than the corresponding signature.
\emph{Linear hashing with separators} \cite{larson1988linear} is a dynamic variant with a linear probe sequence.
\emph{External Robin Hood Hashing} \cite{celia1988external} is similar to linear separator hashing, but it instead pushes out objects that are closest to their respective home address.
For each block, the internal memory stores the smallest distance of its objects to their respective home address.

\myparagraph{Variable-Bit-Length Arrays.}
Variable-bit-length arrays (VLAs) are arrays containing objects of variable size.
VLAs are closely related to PaCHash, which can be used also as VLA by
using the array index instead of the hash function, see \cref{s:variants}.
Conversely, PaCHash can be seen as a VLA where each entry stores a PaCHash bin.
However, most VLAs have some limitations that rule out storing the PaCHash bins efficiently.
A major difference to all VLAs described below is PaCHash allowing objects to span over multiple blocks of fixed size.

Navarro \cite[Section 3.2]{navarro2016compact} describes several
techniques for implementing VLAs.  However, none of them achieves the
same favorable space-time trade-off as the PaCHash VLA.  The closest one --
sampled pointers -- needs $N+n\log(N)/k$ bits of space
 with access
cost bounded by the time needed to skip $k$ objects. Note that this
time can be large when large objects need to be skipped.%
\footnote{Space could be reduced to
  $N+\frac{n}{k}(2+\log\frac{kN}{n})$ bit using Elias-Fano coding of
  the pointers -- resulting in similar space as the PaCHash VLA with $B=kN/n$ but
  with worse access costs.}  All the other described VLAs need several
bits of space overhead per object (multiplied with a factor that depends on the maximum or average
object size).

The VLA introduced by K{\"{u}}lekci~\cite{Kulekci2014VLA} uses wavelet trees~\cite{FerraginaM2000FMIndex} to partition the universe.
This makes the query time depend double logarithmically on the largest element stored in the VLA, a limitation not existing in PaCHash.

Blandford and Blelloch \cite{blandford2008compact} describe dynamic
VLAs and hash tables for variable sized objects.  However, their
technique incurs a constant factor of space overhead and is limited to
objects of bounded size.  They partition the objects into blocks, but
the blocks are generally only partially filled and do not allow
objects crossing block boundaries as in PaCHash.

\section{The PaCHash Data Structure}\label{s:structure}
We now present PaCHash in detail -- a hash table which considerably improves on the data structures from the literature.
It natively supports variable size objects without the need for indirection or empty cells.
It needs only a few bits of internal memory per \emph{block} and still needs only one single I/O operation (of variable length) per query.
PaCHash consists of an \emph{external part} subdivided into $m$ blocks of exactly $\Bs$ bits each that store the actual objects and an \emph{internal part} that allows finding the blocks storing an object.
\Cref{fig:example} gives an example for the external and internal memory data structures.
We deliberately use the word \emph{object} for the stored data because that highlights the flexibility of PaCHash.
Naturally, an object stores a key-value-pair, but it can also store only a value to obtain an external dictionary data structure.
It is even possible to use quotienting by storing the bin index inside the first object of each bin.

\begin{figure}[t]
  \centering
  \includegraphics[width=0.55\textwidth]{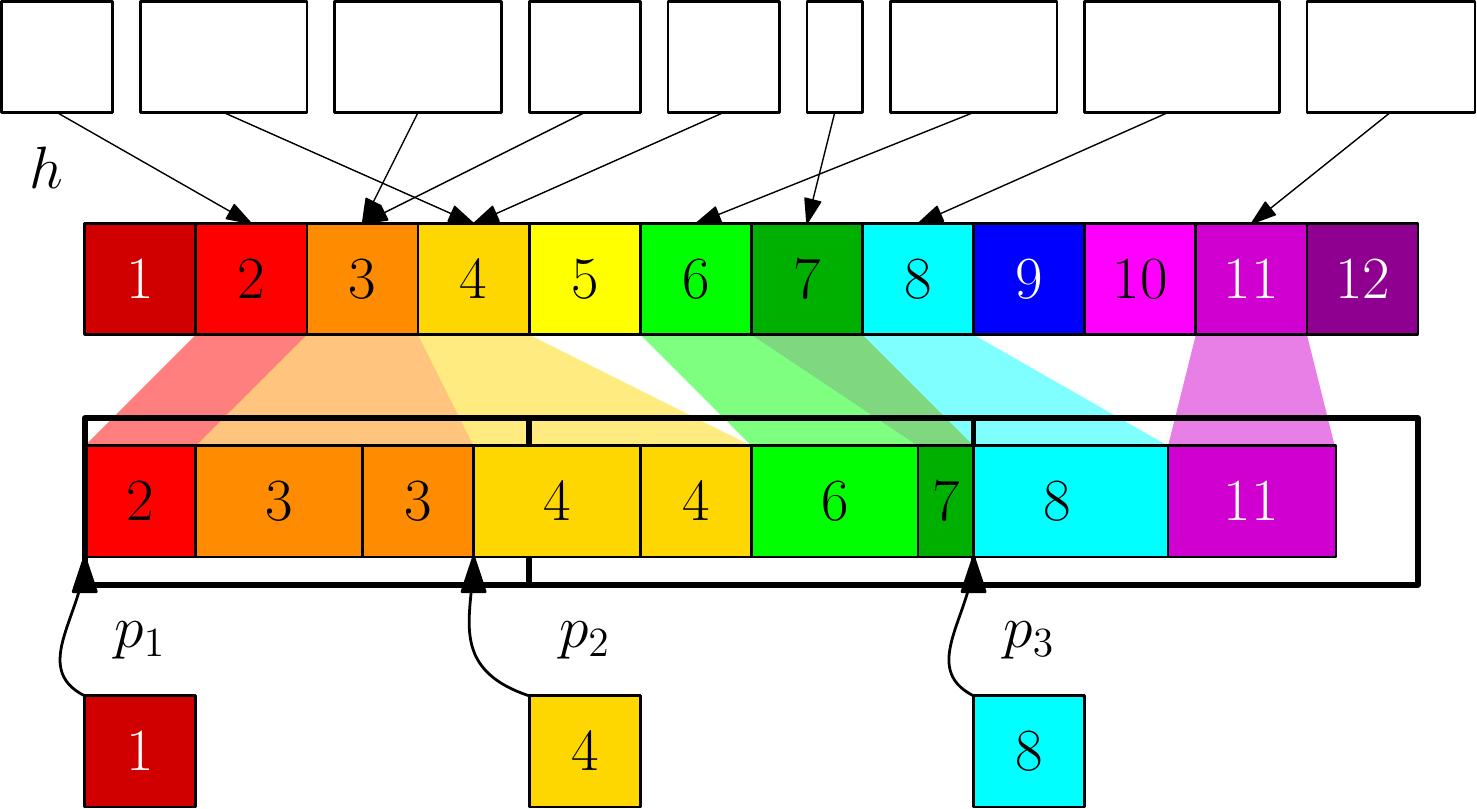}
    \caption{\label{fig:example} Example of PaCHash with $n=9$ objects and $m=3$ blocks. Using the hash function $h$, the objects are mapped to $12$ bins shown as colors, i.e., $a=4$. The bin content is then contiguously written to the external memory blocks. The internal memory index $p$ stores the first bin intersecting with each block. Note that locating bin 8 will return the range $2..3$, i.e., block 2 is loaded superfluously because there is no preceding empty bin that can encode whether it overlaps into the previous block. All other bins are located optimally.}
\end{figure}

\subsection{External Object Representation}
PaCHash stores the objects sorted by a hash function $h$ with a rather small domain, namely $h:K\rightarrow 1..am$, where $K$ is the set of possible keys, $m$ is the number of blocks and $a$ is a tuning parameter that we assume to be a power of two.
The hashes can collide and therefore group the objects into $am$ \emph{bins}.
The objects are now basically stored contiguously.
``Basically'' means that blocks may also contain information needed to find the first object or bin stored in them.
Refer to \cref{ss:encoding} for a discussion of alternative encodings.
Our default assumption is as follows: Each external block stores an offset of size $d=\log\Bs$ bits indicating the bit where the first bin in the block starts.
The remaining space stores the objects contiguously where an object may have an arbitrary size in bits. No space is left between subsequent objects.
In particular, object representations may overlap block boundaries. We assume that objects are encoded in a self-delimiting way, i.e., when we know where an object starts, we
can also find its end. For example, we could have a prefix-free code for the objects.
Construction first sorts the objects by their hash function value.
Then it scans the sorted objects, constructing both the external and the internal data structure along the way.
Refer to \cref{s:analysis} for more details.
If the internal data structure gets lost, for example due to a power outage, it can be re-generated using a single scan over the external memory data.

\subsection{Internal Memory Data Structure}
\label{sec:internalMemory}
Given a bin $b$, the internal memory data structure $p$ can be used to determine a (near\nobreakdash-)minimal range $i..j$ of block indices such that $b$ is stored in that range.
When performing a query, that block range can then be loaded from external memory and scanned for the sought key.
In practice, the resulting latency is often close to that of loading a single block since it includes only one disk seek.
Conceptually, $p$ stores a sequence $\seq{p_1,\ldots,p_m}$ where $p_i$ specifies the first bin whose data is at least partially contained in block~$i$.\footnote{An alternative would be to store the first bin that \emph{starts} in each block. This introduces a special case when a block is fully overlapped by a bin and needs slightly more work when performing queries.}
We can use a predecessor query on $p$ to determine~$i$.
When the predecessor is $b$ itself, we also need to load the previous block.
Another predecessor query or scanning then determines $j$, as illustrated by the pseudocode in \cref{fig:locate}.
To get the most out of this specification, we take empty bins into account: When a bin starts exactly at a block boundary and has an empty predecessor, we store that predecessor.
This implies that if (and only if) a bin $b$ starts at a block boundary and the previous bin $b-1$ is nonempty, retrieving bin $b$ will load one block too much.
Note that $p$ is a monotonically increasing sequence of integers which can be represented with different methods and trade-offs.

\begin{lstlisting}[label=fig:locate,float=t,
  caption={A query for an object $x$ calls $\textrm{locate}(x)$, loads the returned block range, and scans the blocks to find the object content.
  Determining the range boils down to predecessor queries on $p$.}
]
Function locate$(x)$
    $b$ := $h(x)$
    find $i$ such that $p_{i-1} < b \leq p_i$ &\hfill&// &{\footnotesize predecessor query}&
    if $p_i = b$ then $i$ := $i-1$ &\hfill&// &{\footnotesize $b$ may start in previous block}&
    find first $j$ such that $p_j > b$  &\hfill&// &{\footnotesize predecessor query or scan}&
    return $i..(j-1)$
\end{lstlisting}

\myparagraph{Elias-Fano Coding.}
A standard technique for storing monotonic sequences is Elias-Fano coding (see \cref{par:eliasFano}).
A way to interpret the vector $H$ of upper bits of an Elias-Fano coded sequence is that it stores the number of items having each possible combination of most significant bits in unary coding.
To locate the predecessor of item $b = au + \ell$ in the sequence, we 
calculate $\textit{select}_0(u-1)$ on the upper bits $H$, which gives us the start of a cluster of entries that all have most significant bits $u$.
The corresponding index in $L$ can be calculated by subtracting $(u-1)$.
We scan the cluster to find the largest index $i$ with $p_i \leq b$.
In our case, this takes constant expected time (see \cref{lem:analysisQueryWork}).
The internal memory usage is $m(2 + \log(a) + o(1))$ bits (see \cref{ss:analysisSpaceEf}).

\myparagraph{Bit Vector with Succincter.}
\label{sec:bit_vector_with_succincter}
It is also possible to store $p$ as a bit vector with rank and select support.
An item $p_i$ at position $i$ is then represented as a $1$-bit in position $i+p_i$.
The position of the predecessor of a bin $b$ can be found in constant time by calculating $\textit{select}_0(b)-b$.
The actual value can be calculated using a $\textit{select}_1$ query.
Because the bit vector is sparse, we can use Succincter \cite{Patrascu2008Succincter} to compress it and its rank and select structures down to about $m(1.44 + \log(a+1) + o(1))$ bits (see \cref{the:internal_space_requirements_succincter}).

\myparagraph{Entropy Coding.}\label{s:entropyCoding}
We observed that in practice, the bit vector is considerably more regular than a truly random one and thus allows additional compression.
This can be made fast by splitting it into ranges that are compressed individually, e.g., using dictionary compression.
In our experimental evaluation in \cref{s:realWorldDataSets}, we see a space-time trade-off, where we can achieve internal memory space consumption less than the theoretically best results described above in \cref{sec:bit_vector_with_succincter}.

\section{Analysis}\label{s:analysis}
We now formalize the properties of PaCHash in \cref{thm:main} which
basically says the following: External space is just the space needed
to store the variable sized objects plus possibly a few bits per block
to know where the first object in the block starts. Internal space is
about $2+\log a$ bits per block where $a$ is a tuning parameter that
also shows up in a term adding $1/a$ expected I/Os to the retrieval cost.

While proving the theorem, we discuss some variants and implications.
\Cref{ss:construction} considers construction cost and final space consumption, while
\cref{ss:analysisQuery} looks at I/Os and internal work of queries.

\begin{theorem}\label{thm:main}
  Consider $n$ objects of total size $N$ bits which are stored in $m$ blocks of size $\Bs$.
  Let $d\in 0..\log\Bs$ be an encoding-dependent number of bits needed to
  specify where the first bin or object of a block starts
  and $\BsMinus=\Bs-d$ be the payload size per block, i.e., $m=N/\BsMinus$.
  For a parameter $a$, let a random uniform hash function map the objects to $am$ bins.

  Then, PaCHash with Elias-Fano coding needs $m(2+\log a + o(1))$ bits
  of internal memory and $N(1+d/\BsMinus)$ bits of external memory.
  The construction cost is the same as that of sorting the objects using $am$ random integer keys.
  The expected time for retrieving an object of size $|x|$ bits is
  constant plus the time for scanning $1+|x|/\BsMinus+1/a$ blocks.
  The unsuccessful search time is the same except that $|x|$ is replaced by $0$.
\end{theorem}

\subsection{Construction}\label{ss:construction}
Assuming that the set of input objects is stored in compressed form on external memory, we mainly need to sort the objects by their hash function value.
In our model, this has complexity $\sort(N)$.
In most practically relevant situations, this can even be done in $\Oh{\scan(N)}$ using integer sorting, see \cref{s:externalSorting} for details.

The sorted representation is then scanned and basically copied to the output, only adding $d$ bits of information within each block, which allow a query to initialize the scanning operation.
What $d$ is depends on the concrete encoding of the data, ranging from $d=0$ for objects of identical size or for 0-terminated strings to $d=\log(\Bs)$ bits when we explicitly encode the starting position of an object or bin.
Refer to \cref{ss:encoding} for examples.

\begin{lemma}
  \label{lem:ef_construction_pachash}
    When using Elias-Fano coding to store $p$, the index needs $2 + \log a + o(1)$ bits of internal memory per block and can be constructed in time $\Oh{m}$.
  \label{ss:analysisSpaceEf}
\end{lemma}
\begin{proof}
$p$ consists of $k=m$ integers $\leq am=U$.
Inserting this into the space usage of Elias-Fano coded sequences (see \cref{par:eliasFano}) gives us
$\textrm{space}(p) = k(2 + \log(U/k)) + 1 = m(2 + \log(am/m)) + 1 = m(2+\log a) + 1$.
The $\textit{select}_0$ data structure on the upper bits $H$ can be stored in $o(m)$ bits \cite{clark1997compact}.
Each of the $m$ insertions into the sequence can be done in constant time while generating the external object representation.
The construction of the $\textit{select}_0$ data structure takes time $\Oh{m}$.
\end{proof}

\begin{lemma}\label{the:internal_space_requirements_succincter}
    When using Succincter \cite{Patrascu2008Succincter} to store $p$, the index needs $1.4427+\log(a+1)+o(1)$ bits of internal memory per block.
\end{lemma}
\begin{proof}(Sketch, for full proof see \cref{pro:full_proof_internal_space_requirements})
  Using Succincter, i.e., \cite[Theorem 2]{Patrascu2008Succincter} with a length-\((a+1)m\) bit vector containing \(m\) ones, we can represent the internal memory index using only \(\log\binomial{(a+1)m}{m}+o(m)\leq m\left(1.4427+\log(a+1)\right)+o(m)\) bits, which results in the space mentioned above per external memory block.
\end{proof}

The lower bound for the space usage of a minimum $k$-perfect hash function for objects of identical size approaches $n \cdot (\log(e) + \log(k!/k^k)/k)$ \cite{BelazzouguiBD09}.
Using Stirling's approximation, we derive a new lower space bound that is easier to interpret.

\begin{align*}
    n \cdot &(\log(e) + \log(k!/k^k)/k) \\
        & \approx n \cdot \left(\log(e) + \log\left(\frac{\sqrt{2 \pi k} (k/e)^k}{k^k}\right)/k\right)
          = n \cdot \left(\log(e) + \log(\sqrt{2 \pi k} (1/e^k))/k\right) \\
        & = n \cdot \left(\log(e) + \frac{\log(\sqrt{2 \pi k})}{k} - \frac{\log(e^k)}{k}\right)
          = n \cdot \left(\log(e) + \frac{\log((2 \pi k)^{1/2})}{k} - \log(e)\right) \\
        &  = \frac{n}{k} \cdot \frac{1}{2}\log(2 \pi k)
\end{align*}

The value $n/k$ is the number of blocks, so M$k$PHFs need $\Omega(\log k)$ bits of space per block, while we show above that PaCHash needs a constant number.
In a way, PaCHash therefore breaks the theoretical lower space bounds of M$k$PHFs while keeping the same $\Oh{1}$ query time.
Choosing parameter $a$ large can bring the number of I/O operations arbitrarily close to optimal, independently of $k$.

\subsection{Query}\label{ss:analysisQuery}
We first show that a query loads a small expected number of blocks, depending only on the size of that specific object -- not the other objects in the data structure.
We then show that the exact blocks to be loaded can be determined upfront without any I/O operations, using constant time.

\begin{lemma}
\label{lem:blockLoads}
Retrieving an object $x$ of size $|x|$ from a PaCHash data structure loads $\leq 1+|x|/\BsMinus+1/a$ consecutive blocks from the external memory in expectation (setting $|x|=0$ if $x$ is not in the table).%
    \footnote{Using fewer estimates in the proof one can derive a bound of $1+\frac{|x|-c+1-e^{-\beta}}{\BsMinus}+\frac{1}{a}$ where $\beta=\frac{n\BsMinus}{Na}$ is the average number of objects per bin and $c$ is the greatest common divisor of $\BsMinus$ and all object sizes. In particular, for objects of identical size dividing $\Bs$, the bound is close to $1+1/a$.}
\end{lemma}
\begin{proof}
  We first derive the expected number of blocks overlapped by the bin $b_x=h(x)$ that $x$ is stored in.
We then analyze the edge case that PaCHash sometimes loads one additional block unnecessarily even though it is not overlapped.

The expected size $\mathds{E}(|b_x|)$ of $b_x$ is the sum of $|x|$ and all other objects from the input set $\Objects$ that are mapped to it:

\begin{align*}
\mathds{E}(|b_x|)
    &= |x| + \sum_{y \in \Objects, y \not= x}|y| \mathds{P}(y \in b_x) \\
    &\leq |x| + \sum_{y \in \Objects}|y| \mathds{P}(y \in b_x)
    = |x| + \sum_{y \in \Objects}|y| \cdot \frac{1}{am}
    = |x| + \BsMinus m \cdot \frac{1}{am}
    = |x| + \frac{\BsMinus}{a}
\end{align*}

Let $X$ denote the number of blocks overlapped by bin $b_x$.
Assuming that the block boundaries and bin boundaries are
statistically independent,\footnote{We can guarantee the independence
  by cyclically shifting the data structure, i.e., we set the offset
  of the first block to a random number in $0..(\BsMinus-1)$
  and let the last bins wrap around into the first block.}
and using the linearity of the expected value, we get
$\mathds{E}(X) = 1 + (\mathds{E}(|b_x|)-1)/\BsMinus = 1 + |x|/\BsMinus + 1/a - 1/\BsMinus$.

At a position $i$, the sequence $p$ stores the first bin $b_i$ that intersects with block $i$.
Most of the time, this also means that $b_i$ extends into block $i-1$, which is why queries load that block as well.
When a bin starts \emph{exactly} at a block boundary, though, the previous block is not actually needed.
Because bin boundaries are statistically independent of block boundaries, the probability of that happening is $1/\BsMinus$.%
\footnote{
  When the preceding bin $b_{-1}$ is empty, PaCHash stores that empty bin in $p$, as described in \cref{s:structure}.
    This means that the probability of unnecessary block loads actually is smaller, namely $\frac{1}{\BsMinus}(1-\mathds{P}(|b_{-1}| > 0))$,
    where $\mathds{P}(|b_{-1}| > 0) = \left(1-\frac{1}{am}\right)^n\approx e^{-\frac{n}{am}}$ is the probability of $b_{-1}$ being empty.}

We get the result by putting together the expected blocks overlapped by a bin and the probability for loading one single block too much.
For negative queries, we are interested in the size of the bin that $x$ would be hashed to, so we can simply set $|x|=0$.
\end{proof}

\begin{lemma}
  \label{lem:analysisQueryWork}
  When using Elias-Fano coding for the index data structure of PaCHash, the
  range of blocks containing the bin of an
  object $x$ can be found in expected constant time.
\end{lemma}
\begin{proof}
A query for an object $x$ consists of four steps.
First, we hash $x$ to get the corresponding bin $b_x=au+\ell$, where $a$ is the tuning parameter of PaCHash.
We then execute a constant time \cite{clark1997compact} $\textit{select}_0$ query on the upper bits $H$.
That gives us the start of a cluster of entries in the sequence that all have the same $\log(m)$ most significant bits $u$.
We need to iterate over the cluster entries which are $<b_x$ until we find the predecessor.
Each cluster entry corresponds to a stored bin index.
Let us bound the expected size $\mathds{E}(Y_u)$ of all bins that have most significant bits $u$ and are $<b_x$.

\begin{align*}
    \mathds{E}(Y_u)
    &= \sum_{y\in\Objects}|y| \cdot \mathds{P}(h(y) \textrm{ has MSB}=u\textrm{; }h(y)<h(x)) \\
    &\leq \sum_{y\in\Objects}|y| \cdot \mathds{P}(h(y) \textrm{ has MSB}=u)
    =\frac{1}{m}\sum_{y\in\Objects}|y|=\frac{m\BsMinus}{m}=\BsMinus
\end{align*}

The expected number of cluster entries we need to scan is therefore  $\mathds{E}(Y_u)/\BsMinus=1$.
The practical implementation then further scans the cluster to find the last block overlapping $b_x$.
This takes non-constant time $\Oh{1+|x|/\BsMinus}$, which is not a problem since a proportional number of blocks are loaded anyway.
However, we strengthen the lemma by observing that we can also use another $\textit{select}_0$ query followed by a backward scan of the cluster.
\end{proof}

\subsection{Details on External Sorting}
\label{s:externalSorting}
We now show that the external sorting needed during construction of a PaCHash data structure can be done in scanning complexity using very modest additional assumptions. 
First note that the problem of sorting objects during construction is easy when the average object size exceeds the block size, i.e., $N/n>\Bs$ and thus $n<N/\Bs$.
In that case, a variant of bucket sort that maps the keys to $\Oh{n}$ buckets runs with linear internal expected work and $\Oh{n+N/\Bs}=\Oh{N/\Bs}$ I/Os \cite[Theorem~5.9]{SMDD19short}.

On the other hand, the average object size $N/n$ must be at least $\log n$ since we are looking at objects with unique keys.
For the remaining case $\log n\leq N/n\leq \Bs$, we additionally make a \emph{tall cache assumption} quite usual for external memory \cite{Frigo:1999:COA} where $\Ms>\Bs^2$. Since the index data structure has at least $N/\Bs$ bits, we also know that
$\Ms\geq N/\Bs$.
A single scan of the input can partition it into pieces of size about
\(\frac{N}{\Ms/\Bs}\leq \frac{N}{(N/\Bs)/\Bs}=\Bs^2\leq \Ms\)
which fit into internal memory. Moreover, since the average object size is $\geq \log n$, we can afford
to replace the objects in an internally sorted fragment of the input
by key-pointer pairs which once more allows us to use bucket sort --
this time running in internal memory.

\section{Variants and Refinements}\label{s:variants}
Up until now, PaCHash was described as a static, external hash table for objects of variable size.
The following section describes variants of the scheme.

\myparagraph{Object Encoding.}
\label{ss:encoding}
Instead of storing objects contiguously with a self-delimiting encoding, PaCHash allows for a wide range of other options, as shown in \Cref{tab:d}.
In general, we have a trade-off between the space needed to decode the objects in a block and the strength of assumptions made on object representation.
For example, explicitly storing the offsets of objects in blocks removes the restriction to a self-delimiting encoding, without increasing the size of the internal data structure.
Another important case are objects of identical size
where we can calculate the block offset at query time and therefore need no external space overhead.
When the object size divides the block size, it can be shown that the expected number of I/O operations is close to $1+1/a$.

\begin{table}
    \caption{External space overhead of $d$ bits per block in order to facilitate scanning that block. The term $+1$ when $d\not=0$ is needed for the case that no object starts in a block.}
    \label{tab:d}
    \centering
  \begin{tabularx}{\columnwidth}{l X}
    \toprule
    $d$ & Case Description\\ \midrule
    0 & Identical object sizes, zero terminated strings and analogous cases\\
    $\lceil\log(w+1)\rceil$ & Objects that use variable bit-length encoding with $\leq w \leq \Bs$ bits\\
    $\lceil\log(W/w+1)\rceil$ & Objects of size divisible by $w$ with $W=\min(\Bs, \max\text{ object size})$\\
    $\lceil\log(\Bs)\rceil$ & Explicit storage of a starting position of a bin\\
    \bottomrule
  \end{tabularx}
\end{table}

\myparagraph{Memory Locations.}
PaCHash can be stored fully externally.
By doing so, the number of I/Os for a query is increased by three (two I/Os to query the rank and select data structure on the bit vector of the Elias-Fano coding and one I/O to get the remaining bits).
The number of I/Os can be reduced by interleaving the arrays of the Elias-Fano coding.
PaCHash is also interesting as a purely internal data structure since it allows for configurations that need less space than any previous approach, even for objects of identical size.
A variant that simplifies the external memory representation is to store the $d$ bits of offsets per block in an internal memory data structure, possibly interleaved with the Elias-Fano representation.
A variant enabling faster scanning of blocks separates keys and values \cite{lu2017wisckey}, for example by storing $\log\Bs$ bits of offset for each object.

\myparagraph{Functional Enhancements.}
Because PaCHash sorts objects by their hashed key, \emph{range queries} with respect to the original keys are not immediately possible.
Litwin and Lomet \cite{litwin1986bounded} implement range queries for hash tables by partitioning the key space into smaller pieces.
An index tree then leads to a number of small (PaCHash) tables that are fully scanned.
Order-preserving hash functions \cite{garg1986order} are another alternative.
PaCHash can be made \emph{dynamic} using standard techniques like a Log-Structured Merge Tree \cite{o1996log,luo2020lsm}.
Merging multiple PaCHash data structures is possible efficiently.
The idea is to construct the hash function $h$ by first hashing to a larger range and then mapping it linearly to the range $am$.
When updating $h$ to the new total number of blocks, the objects of both input data structures are already sorted and can be merged with a linear sweep.

\myparagraph{PaCHash as Variable-Bit-Length Array.}
Since one of PaCHash's key features is to store objects of variable size efficiently, it can also be used as variable-bit-length array.
To this end, we simply use the array index as hash function if we also store the number of previously stored objects.
However, we then have to assume that objects stored in the PaCHash VLA are self-delimiting, as this allows us to identify the objects within a block.
Note that this assumption is satisfied in a lot of applications VLAs are used in, e.g., when storing variable length codes like Elias-\(\gamma\) and -\(\delta\) codes~\cite{Elias74} or Golomb codes~\cite{Golomb1966Encodings}.
Alternatively, in external memory, we can lift the restriction to self-delimiting objects by storing offsets as described above.
The number of previously stored objects is necessary to identify the element within the block, and requires at most \(\ceil{\log n}\) bits per external memory block.

\definecolor{colorPaCHash}{HTML}{377EB8}
\definecolor{colorSeparator}{HTML}{000000}
\definecolor{colorPthash}{HTML}{000000}
\definecolor{colorLevelDb}{HTML}{FF7F00}
\definecolor{colorSilt}{HTML}{4DAF4A}
\definecolor{colorRocksDb}{HTML}{984EA3}
\definecolor{colorRecSplit}{HTML}{A65628}
\definecolor{colorUnorderedMap}{HTML}{F781BF}
\definecolor{colorCuckoo}{HTML}{E41A1C}
\definecolor{colorChd}{HTML}{444444}

\pgfplotscreateplotcyclelist{mycolorlist}{%
  {colorPaCHash, mark=diamond},
  {colorLevelDb, mark=square},
  {colorSilt, mark=o},
  {colorSeparator, mark=triangle},
  {colorRocksDb, mark=pentagon}
}
\pgfdeclareplotmark{pacman}{%
  \pgfpathmoveto{\pgfpointorigin}%
  \pgfpathlineto{\pgfqpointpolar{40}{1.2\pgfplotmarksize}}%
  \pgfpatharc{40}{320}{1.2\pgfplotmarksize}%
  \pgfpathlineto{\pgfpointorigin}%
  \pgfusepath{fill}
}
\pgfdeclareplotmark{flippedTriangle}{%
  \pgfpathmoveto{\pgfqpointpolar{-90}{1.2\pgfplotmarksize}}%
  \pgfpathlineto{\pgfqpointpolar{30}{1.2\pgfplotmarksize}}%
  \pgfpathlineto{\pgfqpointpolar{150}{1.2\pgfplotmarksize}}%
  \pgfpathclose%
  \pgfusepath{stroke}
}

\begin{figure*}[t]
            \centering
        \begin{tikzpicture}
            \begin{axis}[
              title={},
              plotIoVolume,
                xlabel={Average object size},
                ylabel={\begin{tabular}{c}\textbf{I/O Volume}\\average [B/Query]\end{tabular}},
                ytick distance=1024,
                xtick distance=128,
                scaled ticks=false,
                cycle list name=mycolorlist,
            ]
            \addplot+[mark=none,densely dotted] coordinates { (64,6208.0) (960,7104.0) };
            \addlegendentry{$a$=2 theory};
            \addplot+[mark=none,densely dotted] coordinates { (64,5184.0) (960,6080.0) };
            \addlegendentry{$a$=4 theory};
            \addplot+[mark=none,densely dotted] coordinates { (64,4672.0) (960,5568.0) };
            \addlegendentry{$a$=8 theory};
            \addplot+[mark=none,densely dotted] coordinates { (64,4416.0) (960,5312.0) };
            \addlegendentry{$a$=16 theory};
            \addplot+[mark=none,densely dotted] coordinates { (64,4288.0) (960,5184.0) };
            \addlegendentry{$a$=32 theory};

            \addplot+[only marks] coordinates { (64,6217.65) (128,6280.56) (192,6344.54) (256,6413.72) (320,6478.19) (384,6541.11) (448,6606.11) (512,6663.78) (576,6728.29) (640,6795.92) (704,6857.24) (768,6920.68) (832,6987.776) (896,7047.54) (960,7110.45) };
            \addlegendentry{$a$=2 real};
            \addplot+[only marks] coordinates { (64,5195.37) (128,5258.28) (192,5320.66) (256,5384.11) (320,5450.67) (384,5513.58) (448,5578.59) (512,5644.66) (576,5705.97) (640,5769.42) (704,5830.25) (768,5895.78) (832,5958.656) (896,6024.192) (960,6086.08) };
            \addlegendentry{$a$=4 real};
            \addplot+[only marks] coordinates { (64,4683.16) (128,4742.92) (192,4810.55) (256,4876.08) (320,4937.93) (384,4999.29) (448,5062.74) (512,5128.77) (576,5192.21) (640,5249.35) (704,5316.44) (768,5386.73) (832,5445.96) (896,5510.96) (960,5572.28) };
            \addlegendentry{$a$=8 real};
            \addplot+[only marks] coordinates { (64,4425.24) (128,4489.18) (192,4548.94) (256,4613.45) (320,4676.85) (384,4739.28) (448,4802.68) (512,4868.22) (576,4929.58) (640,4996.14) (704,5057.99) (768,5123.52) (832,5186.97) (896,5256.192) (960,5316.44) };
            \addlegendentry{$a$=16 real};
            \addplot+[only marks] coordinates { (64,4295.72) (128,4358.64) (192,4419.99) (256,4483.93) (320,4547.91) (384,4612.38) (448,4678.98) (512,4738.21) (576,4802.15) (640,4865.06) (704,4932.2) (768,4993.52) (832,5056.43) (896,5117.79) (960,5184.35) };
            \addlegendentry{$a$=32 real};

            \legend{};
            \end{axis}
        \end{tikzpicture}
        \hfill
        \begin{tikzpicture}
            \begin{axis}[
              title={},
              plotIoVolume,
                xlabel={Average object size},
                ylabel={\begin{tabular}{c}\textbf{Query Time}\\direct I/O [$\mu$s/Query]\end{tabular}},
                scaled ticks=false,
                xtick distance=128,
                legend columns=5,
                legend to name=bytesFetchedLegend,
                cycle list name=mycolorlist,
            ]
            \addplot coordinates { (64,1.92733) (128,1.93111) (192,1.95664) (256,1.98505) (320,2.00797) (384,2.03605) (448,2.0551) (512,2.07464) (576,2.09549) (640,2.12093) (704,2.14085) (768,2.15707) (832,2.17782) (896,2.19266) (960,2.21397) };
            \addlegendentry{$a$=2};
            \addplot coordinates { (64,1.52544) (128,1.53117) (192,1.55553) (256,1.57366) (320,1.5992) (384,1.62663) (448,1.65243) (512,1.67848) (576,1.70195) (640,1.72942) (704,1.75093) (768,1.77753) (832,1.80446) (896,1.82539) (960,1.8529) };
            \addlegendentry{$a$=4};
            \addplot coordinates { (64,1.44709) (128,1.41284) (192,1.40966) (256,1.41658) (320,1.42236) (384,1.44536) (448,1.46185) (512,1.48102) (576,1.5011) (640,1.52612) (704,1.54811) (768,1.57611) (832,1.60159) (896,1.66289) (960,1.65146) };
            \addlegendentry{$a$=8};
            \addplot coordinates { (64,1.41947) (128,1.37727) (192,1.38464) (256,1.37738) (320,1.39462) (384,1.3877) (448,1.39832) (512,1.40442) (576,1.42248) (640,1.43931) (704,1.45523) (768,1.47848) (832,1.50064) (896,1.52706) (960,1.55496) };
            \addlegendentry{$a$=16};
            \addplot coordinates { (64,1.40436) (128,1.36983) (192,1.36934) (256,1.36936) (320,1.36894) (384,1.38835) (448,1.38372) (512,1.40011) (576,1.40596) (640,1.41542) (704,1.42508) (768,1.44638) (832,1.45731) (896,1.47854) (960,1.50287) };
            \addlegendentry{$a$=32};

          \end{axis}
          \end{tikzpicture}

          \begin{tikzpicture}
            \ref*{bytesFetchedLegend}
          \end{tikzpicture}
    \caption{Dependence of I/O volume and query time on the average object size $s$. Sizes are normal distributed with variance $s/5$, rounded to the next positive integer. Dotted lines show theoretic I/O volumes, while marks show measurements. Note that the measurements closely match the analysis. Using other distributions and plotting over the returned objects' sizes gives equivalent results.}
    \label{fig:bytesFetched}
\end{figure*}
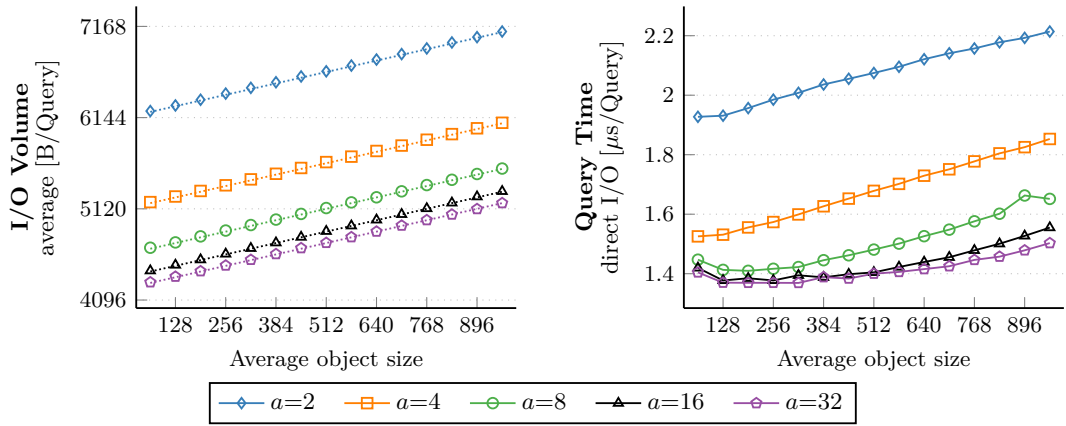

\section{Experiments}\label{s:experiments}
The code and scripts needed to reproduce our experiments are available on GitHub under the General Public License:
\url{https://github.com/ByteHamster/PaCHash}.
The code for the comparison with competitors (including our competitors' code with some patches) is available on GitHub as well:
\url{https://github.com/ByteHamster/PaCHash-Experiments}.
The latter repository also contains a Docker image that can build and run a simplified version of the experiments from \cref{fig:bytesFetched,fig:competitorsIdentical,fig:competitorsVariable} in about 30 minutes.

\myparagraph{Experimental Setup.}
We run our experiments on an Intel i7 11700 processor with 8 cores and a base clock speed of 2.5 GHz.
We use a Samsung 980 Pro NVMe SSD with a capacity of 1 TB.
The machine runs Ubuntu 21.10 with Linux 5.13.0.
We use the GNU C++ compiler version 11.2.0 with optimization flags \texttt{-O3 -march=native}.
Externally, each block of size $\Bs=2^{15}$ bits (4096 bytes) stores a table of 8 byte keys and 2 byte object offsets.
During construction, we sort pointers to the objects using IPS\textsuperscript{2}Ra \cite{axtmann2020engineering}.
Unless otherwise specified, the index is an Elias-Fano coded sequence based on sdsl's \cite{gog2014theory} arrays of flexible bit width and the select data structures by Kurpicz \cite{kurpicz2022pasta}.
For the I/O operations, we use \texttt{io\_uring}.
Query operations keep a queue of 128 asynchronous requests in flight.

\myparagraph{Competitors.}
To our knowledge, there is no existing implementation of a hash table for variable size objects that is simultaneously aimed at low internal memory usage and few I/O operations.
As the main competitors, we choose LevelDB \cite{google2021leveldb}, RocksDB \cite{facebook2021rocksdb}, and SILT \cite{lim2011silt}.
To abstract from the different implementations of I/O operations, we also extract the internal memory index (address calculation) from some competitors.
Additionally, we compare PaCHash to std::unordered\_map, as well as the perfect hash functions RecSplit \cite{esposito2020recsplit}, CHD \cite{BelazzouguiBD09,cmph}, and PTHash \cite{pibiri2021pthash}.%
Despite std::unordered\_map not being tuned for efficiency, it is a widely available, general purpose hash table that can be seen as baseline for the simple idea of explicitly storing pointers instead of building a compressed index data structure.\footnote{In this setting, general purpose internal memory hash tables do not work well, as they introduce an overhead of at least \(\log m\) bits per element to store the positions, and they also have to store the length of the element.}

\label{sec:variableSizeAdaption}
We also implement Separator Hashing \cite{gonnet1988external, larson1984file} and Cuckoo Hashing \cite{azar1994balanced,pagh2003basic}.
In contrast to the original papers, our implementations can be used with objects of variable size $\leq \Bs$ when setting the load factor low enough.
Note that decreasing the load factor increases the number of blocks and therefore the space needed for indexing.
The construction of PaCHash always succeeds, while it can fail for Separator and Cuckoo Hashing depending on the preselected load factor or tuning parameter.
Refer to \cref{fig:separatorLoadFactor} for details.

\begin{table}[t]
  \centering
  \caption{Average internal space usage and average query time for different values of parameter \(a\) and normal distributed object sizes.
    For more information on the query time, which is influenced by the object size, see \cref{fig:bytesFetched}.
    Note that the internal space usage does not depend on the object size.}
  \label{tab:sizeQueryTimeTradeoff}
  \begin{tabular}{lcc}
    \toprule
    \(a\) & \makecell{avg. internal space\\ {[B/block]}} & \makecell{avg. query time\\ {[\(\mu s\)/query]}}\\
    \midrule
     2 & 3.01 & 2.07 \\
     4 & 4.01 & 1.68 \\
     8 & 5.01 & 1.50 \\
    16 & 6.01 & 1.43 \\
    32 & 7.01 & 1.41 \\
    \bottomrule
  \end{tabular}
\end{table}

\begin{figure*}[t]
    \centering
  \begin{subfigure}[b]{.60\textwidth}
    \begin{tabular}{ l r r r }
      \toprule
      &          Twitter &        UniRef 50 &        Wikipedia \\ \midrule
      Objects $n$    &       20 238 968 &       48 531 431 &       16 181 427 \\
      Average size   &            115 B &            281 B &           1731 B \\
      Median size    &             94 B &            194 B &             77 B \\
      Maximum size   &            560 B &            45 KB &           272 KB \\
      Total size $N$ &           2.4 GB &          13.2 GB &          26.3 GB \\
      Objects $>\Bs$ &              0\% &           0.08\% &             12\% \\
      \bottomrule
    \end{tabular}
    \caption{Twitter, UniRef, and Wikipedia real world data sets we use for benchmarks. The median of $77$ bytes of the Wikipedia data set is caused by pages that are redirects.}
    \label{tab:dataSets}
  \end{subfigure}
  \hfill
  \begin{subfigure}[b]{.39\textwidth}
    \centering
    \rotatebox[origin=c]{90}{\small\hspace{9mm}Relative occurrences}
    \parbox{0.8\columnwidth}{
    \begin{tikzpicture}
      \begin{axis}[
        plotSmallHist,
        /pgfplots/every axis plot post/.append style={mark=none},
        yticklabels={,,},
        xmax=440,
        scaled ticks=false,
        title style={at={(1,.5)},anchor=east},
        title={\footnotesize Twitter},
        ]
        \addplot[color=colorPaCHash] coordinates { (4.5,170760) (14.5,711075) (24.5,1278400) (34.5,1487837) (44.5,1488007) (54.5,1296623) (64.5,1216257) (74.5,1103436) (84.5,990193) (94.5,864417) (104.5,773366) (114.5,671412) (124.5,595222) (134.5,590511) (144.5,2680278) (154.5,751562) (164.5,319671) (174.5,206674) (184.5,170264) (194.5,142955) (204.5,131935) (214.5,150733) (224.5,152476) (234.5,212040) (244.5,212915) (254.5,111576) (264.5,101624) (274.5,103521) (284.5,104875) (294.5,115467) (304.5,133020) (314.5,145751) (324.5,145753) (334.5,139555) (344.5,125272) (354.5,122947) (364.5,144979) (374.5,190721) (384.5,147007) (394.5,25618) (404.5,4652) (414.5,4719) (424.5,1659) (434.5,251) (444.5,188) (454.5,132) (464.5,100) (474.5,44) (484.5,122) (494.5,98) };
      \end{axis}
    \end{tikzpicture}
    \begin{tikzpicture}
      \begin{axis}[
        plotSmallHist,
        /pgfplots/every axis plot post/.append style={mark=none},
        yticklabels={,,},
        xmax=4400,
        scaled ticks=false,
        title style={at={(1,.5)},anchor=east},
        title={\footnotesize UniRef 50},
        ]
        \addplot[color=colorPaCHash] coordinates { (13.0,30692) (20.0,109334) (30.0,384287) (42.0,1021677) (56.0,2356726) (72.0,3422884) (90.0,3369252) (110.0,3627489) (132.0,3551402) (156.0,3421956) (182.0,3160322) (210.0,2958077) (240.0,2751979) (272.0,2538286) (306.0,2341860) (342.0,2088225) (380.0,1828191) (420.0,1575247) (462.0,1298834) (506.0,1077691) (552.0,890734) (600.0,743504) (650.0,622248) (702.0,522978) (756.0,441591) (812.0,375591) (870.0,307868) (930.0,260719) (992.0,217704) (1056.0,189223) (1122.0,158241) (1190.0,130457) (1260.0,106350) (1332.0,89589) (1406.0,75621) (1482.0,64244) (1560.0,53685) (1640.0,46277) (1722.0,39507) (1806.0,34069) (1892.0,29026) (1980.0,25354) (2070.0,22340) (2162.0,19699) (2256.0,16290) (2352.0,14410) (2450.0,12542) (2550.0,11316) (2652.0,9655) (2756.0,8364) (2862.0,7682) (2970.0,6618) (3080.0,6024) (3192.0,5425) (3306.0,4628) (3422.0,4231) (3540.0,3739) (3660.0,3387) (3782.0,2901) (3906.0,2672) (4032.0,2554) (4160.0,2276) (4290.0,2048) (4422.0,1821) (4494.0,160) };
      \end{axis}
    \end{tikzpicture}
    \begin{tikzpicture}
      \begin{axis}[
        plotSmallHist,
        xlabel={Object size},
        /pgfplots/every axis plot post/.append style={mark=none},
        yticklabels={,,},
        xmax=44000,
        scaled ticks=false,
        title style={at={(1,.5)}, anchor=east},
        title={\footnotesize Wikipedia},
        ]
        \addplot[color=colorPaCHash] coordinates { (11.0,0) (12.0,0) (13.5,0) (15.0,25) (16.5,205) (19.0,58561) (21.5,145440) (24.0,410300) (27.0,650406) (30.0,818813) (33.5,971115) (38.0,913868) (43.0,723871) (48.0,652354) (54.0,818020) (60.5,671108) (67.5,818757) (76.0,659992) (85.5,414486) (96.0,269741) (108.0,211479) (121.0,137548) (135.5,125755) (152.5,103946) (171.5,77331) (192.5,52406) (216.0,46186) (242.0,52590) (271.5,51262) (305.0,35027) (342.5,35098) (384.5,34672) (431.5,41218) (484.0,42609) (543.0,54607) (610.0,73974) (685.0,92077) (768.5,107532) (862.5,143895) (968.0,186238) (1086.5,233355) (1220.0,279754) (1369.5,272249) (1537.0,286527) (1725.0,297619) (1936.0,300484) (2173.0,309063) (2439.5,312079) (2738.5,307600) (3073.5,300709) (3450.0,285699) (3872.5,266904) (4346.5,250537) (4879.0,226252) (5476.5,205618) (6147.0,182245) (6900.0,160761) (7745.0,137539) (8693.5,117406) (9758.5,100605) (10953.5,85157) (12294.5,71780) (13800.0,59715) (15490.0,50358) (17387.0,42402) (19516.5,35317) (21906.5,28962) (24589.0,23776) (27600.0,19166) (30980.0,15101) (34774.0,11551) (39033.0,10244) (43142.5,6660) };
      \end{axis}
    \end{tikzpicture}}
  \caption{Relative occurrences of object sizes in the real world data sets described in \cref{tab:dataSets}.}
  \label{fig:dataSetsHist}
  \end{subfigure}\\

  \begin{subfigure}[b]{\textwidth}
    \begin{tikzpicture}
      \begin{axis}[
        plotPaCHashComparison,
        title={},
        xlabel={Parameter $a$},
        ylabel={\begin{tabular}{c}\textbf{Space}\\internal [B/Block]\end{tabular}},
        xmode=log,
        log ticks with fixed point,
        log basis x={2},
        ]
        \addplot[color=colorPaCHash,mark=o] coordinates { (1,1.03263) (2,1.6581) (4,2.48241) (8,3.60626) (16,4.80022) (32,6.32515) (64,7.99564) (128,10.584) };
        \addlegendentry{Entropy coded, Twitter};
        \addplot[color=colorPaCHash,mark=star] coordinates { (1,1.32286) (2,2.06929) (4,3.04209) (8,4.1003) (16,5.38882) (32,6.74366) (64,8.47675) (128,10.764) };
        \addlegendentry{Entropy coded, UniRef};
        \addplot[color=colorPaCHash,mark=flippedTriangle] coordinates { (1,2.02006) (2,2.70954) (4,3.44051) (8,4.19718) (16,5.06612) (32,6.3088) (64,7.76899) (128,9.98633) };
        \addlegendentry{Entropy coded, Wikipedia};
        \addplot[color=colorLevelDb,mark=square] coordinates { (1,2.00933) (2,3.00939) (4,4.00935) (8,5.00941) (16,6.00936) (32,7.00942) (64,8.00937) (128,9.00943) };
        \addlegendentry{Elias-Fano, Twitter};

        \addplot[color=colorSeparator,mark=+] coordinates { (1,2.0) (2,2.75489) (4,3.60964) (8,4.52933) (16,5.48687) (32,6.46501) (64,7.45391) (128,8.44832) };
        \addlegendentry{Succincter};

        \legend{};
      \end{axis}
    \end{tikzpicture}
    \hfill
    \begin{tikzpicture}
      \begin{axis}[
        plotPaCHashComparison,
        xlabel={Parameter $a$},
        ylabel={\begin{tabular}{c}\textbf{Query Throughput}\\direct I/O [kQueries/s]\end{tabular}},
        xmode=log,
        log ticks with fixed point,
        log basis x={2},
        ]
        \addplot[color=colorPaCHash,mark=o] coordinates { (1,61.1479) (2,82.6537) (4,88.0008) (8,111.825) (16,164.736) (32,211.795) (64,293.69) (128,372.189) };
        \addlegendentry{Entropy coded, Twitter};
        \addplot[color=colorPaCHash,mark=star] coordinates { (1,56.2068) (2,54.831) (4,80.3118) (8,90.0553) (16,113.96) (32,185.515) (64,255.961) (128,342.37) };
        \addlegendentry{Entropy coded, UniRef};
        \addplot[color=colorPaCHash,mark=flippedTriangle] coordinates { (1,73.0805) (2,48.9809) (4,50.8569) (8,60.232) (16,104.945) (32,147.159) (64,224.837) (128,302.212) };
        \addlegendentry{Entropy coded, Wikipedia};
        \addplot[color=colorLevelDb,mark=square] coordinates { (1,263.135) (2,377.074) (4,483.637) (8,563.804) (16,615.638) (32,645.162) (64,661.376) (128,670.092) };
        \addlegendentry{Elias-Fano, Twitter};
        \addplot[color=colorLevelDb,mark=triangle] coordinates { (1,269.372) (2,371.425) (4,468.823) (8,541.81) (16,588.698) (32,615.89) (64,630.782) (128,638.164) };
        \addlegendentry{Elias-Fano, UniRef};
        \addplot[color=colorLevelDb,mark=diamond] coordinates { (1,303.552) (2,370.188) (4,416.667) (8,445.368) (16,463.034) (32,471.921) (64,476.948) (128,479.386) };
        \addlegendentry{Elias-Fano, Wikipedia};

        \legend{};
      \end{axis}
    \end{tikzpicture}

    \begin{tikzpicture}
      \begin{axis}[
        height=0.5cm,
        width=5cm,
        hide axis,
        xmin=10,
        xmax=50,
        ymin=0,
        ymax=0.4,
        legend style={font=\small},
        legend columns=3,
        ]
        \addplot[color=colorPaCHash,mark=o] coordinates { (0,0) };
        \addlegendentry{Entropy coded, Twitter};
        \addplot[color=colorPaCHash,mark=star] coordinates { (0,0) };
        \addlegendentry{Entropy coded, UniRef};
        \addplot[color=colorPaCHash,mark=flippedTriangle] coordinates { (0,0) };
        \addlegendentry{Entropy coded, Wikipedia};
        \addplot[color=colorLevelDb,mark=square] coordinates { (0,0) };
        \addlegendentry{Elias-Fano, Twitter};
        \addplot[color=colorLevelDb,mark=triangle] coordinates { (0,0) };
        \addlegendentry{Elias-Fano, UniRef};
        \addplot[color=colorLevelDb,mark=diamond] coordinates { (0,0) };
        \addlegendentry{Elias-Fano, Wikipedia};

        \addplot[color=colorSeparator,mark=+] coordinates { (0,0) };
        \addlegendentry{Succincter};

      \end{axis}
    \end{tikzpicture}
    \caption{PaCHash with real world data sets using different index data structures. There is no practical implementation of Succincter \cite{Patrascu2008Succincter}, so we only give calculated values and no throughput. The space usage of Elias-Fano and Succincter is independent of the object size distribution, so we plot only one data set.}
    \label{fig:differentPaCHashConfigs}
  \end{subfigure}
  \caption{Space and query throughput of PaCHash with real world data sets.}
\end{figure*}
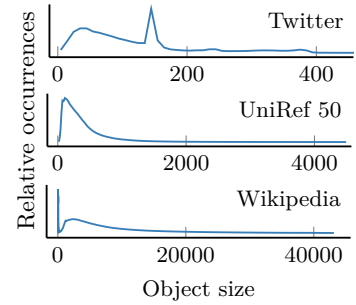
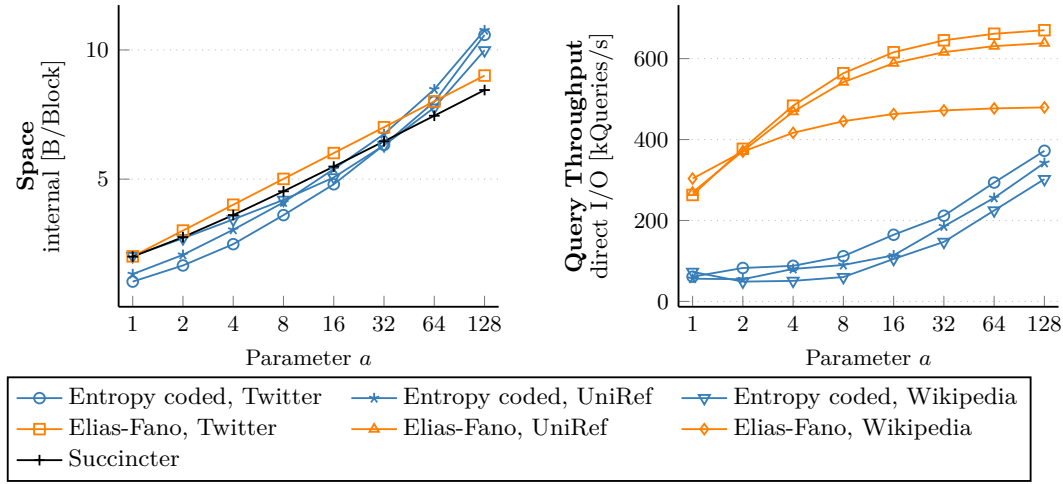

\subsection{PaCHash Configurations}
The parameter $a$ provides a trade-off between internal space usage and query performance, see \cref{tab:sizeQueryTimeTradeoff}.
\Cref{fig:bytesFetched} plots the bytes read per query, depending on the average object size and parameter $a$.
It confirms the results of our theoretical analysis in practice.
The throughput of the Elias-Fano representation increases when parameter $a$ gets larger because the SSD needs to load fewer blocks.
We also see that (at least for larger $a$) query times grow more slowly with object size than the I/O volume.
We choose $a=8$ for the comparison with competitors because it achieves a good balance between space usage ($\approx 5$~bits/block) and throughput ($\approx 700$k Queries/second).

\subsection{PaCHash with Real World Data Sets}
\label{s:realWorldDataSets}
\Cref{fig:differentPaCHashConfigs} compares throughput and space usage of PaCHash using real world size distributions and different index data structures.
The Twitter data set contains tweets from 01.08.--05.08.2021 and has only small objects.
The UniRef~50 protein database \cite{suzek2007uniref} contains some objects larger than the block size
and the LZ4 compressed \cite{collet2011lz4GitHub} English Wikipedia from November 2021 contains significantly larger objects.
See \cref{tab:dataSets,fig:dataSetsHist} for details.

The entropy coded bit vector saves up to one bit of internal memory per block for small $a$.
While it comes with a performance penalty caused by decompression (up to eight times slower than Elias-Fano), it is fast enough that it can be useful for some applications.
Succincter provides space usage lower than Elias-Fano but has no implementation.
Note that for \(a\leq 16\), the entropy coded bit vector requires even less space than succincter.
Only for \(a\geq 64\) it requires more space than Elias-Fano.

\subsection{Comparison with Competitors}
We compare PaCHash to other hash table data structures -- see \Cref{tab:competitorConfigurations} for the exact configurations used. 
\Cref{fig:competitorsIdentical} shows measurements for identical size objects in order to allow for a large set of competitors.
\Cref{fig:competitorsVariable} shows measurements for variable size objects containing fewer data points due to the lack of support for variable size objects by most competitors.
Perhaps the closest contender to PaCHash is the Separator method where our implementation partially allows variable object size.
It needs comparable internal space and has faster queries (always a single block access).
However, Separator not only has slower construction, but it also cannot achieve a load factor close to 100\% except for objects with identical size when the block size is divisible by the object size.
\Cref{fig:separatorLoadFactor} gives details showing load factors between 85\% and 95\% in typical cases.

The perfect hashing methods CHD and RecSplit have similar problems with respect to variable size objects and are more expensive with respect to internal space and construction costs.
While PTHash offers fast construction and queries, it does not support variable size objects and needs more internal space.
Cuckoo hashing needs no internal space but has more expensive queries and problems with variable size objects, like Separator or perfect hashing.

The object stores LevelDB, RocksDB, and SILT have much larger internal space requirements \emph{and} some external overhead.
In part this comparison is unfair since they have additional functionality like dynamic operation.
For SILT and LevelDB we have been able to extract the static part but still get considerably more space and lower performance than PaCHash.
\Cref{fig:competitorsIdentical,fig:competitorsVariable} contain measurements for both the full competitors and their static parts, so the overhead originating from dynamic operation can be read off them.
Comparing query throughput is complicated because of different file access modes, internal caching, and history dependent performance for the actual SSD accesses (the controller uses caching and rearranges data outside the control of the user).
We have therefore looked at two different access methods and also at only the index data structure.
However, overall, we get a consistent picture with Separator being the fastest method followed by PaCHash.
A comparison with the internal hash table std::unordered\_map is also instructive.
We naturally get faster construction and high internal space consumption.
Surprisingly, access to the internal data structure is only faster than PaCHash for very small inputs that fit into cache. 

While not as surprisingly, it should be noted that all object stores supporting variable size objects do not show any difference with respect to (internal and/or external) space requirements, construction and query throughput when storing variable size objects compared to identical size objects.
Thus, all benefits of PaCHash described above hold true for variable size objects as well.

\section{Conclusion and Future Work}\label{s:conclusion}
With  PaCHash, we present a static hash table
that can space-efficiently store variable size (possibly compressed) objects.
The objects are stored contiguously without the usual need for empty space to equalize the nonuniformity in assignment by a hash function. This is facilitated by an index data structure that needs only a constant number of internal memory bits per external memory block. In constant expected time, it yields a near-optimal range of blocks that contain the sought object.
Our implementation of PaCHash considerably outperforms previous object stores for variable size objects and even matches or outperforms systems that are purely internal memory or only handle objects of identical size.

\begin{figure*}[p]
  \input{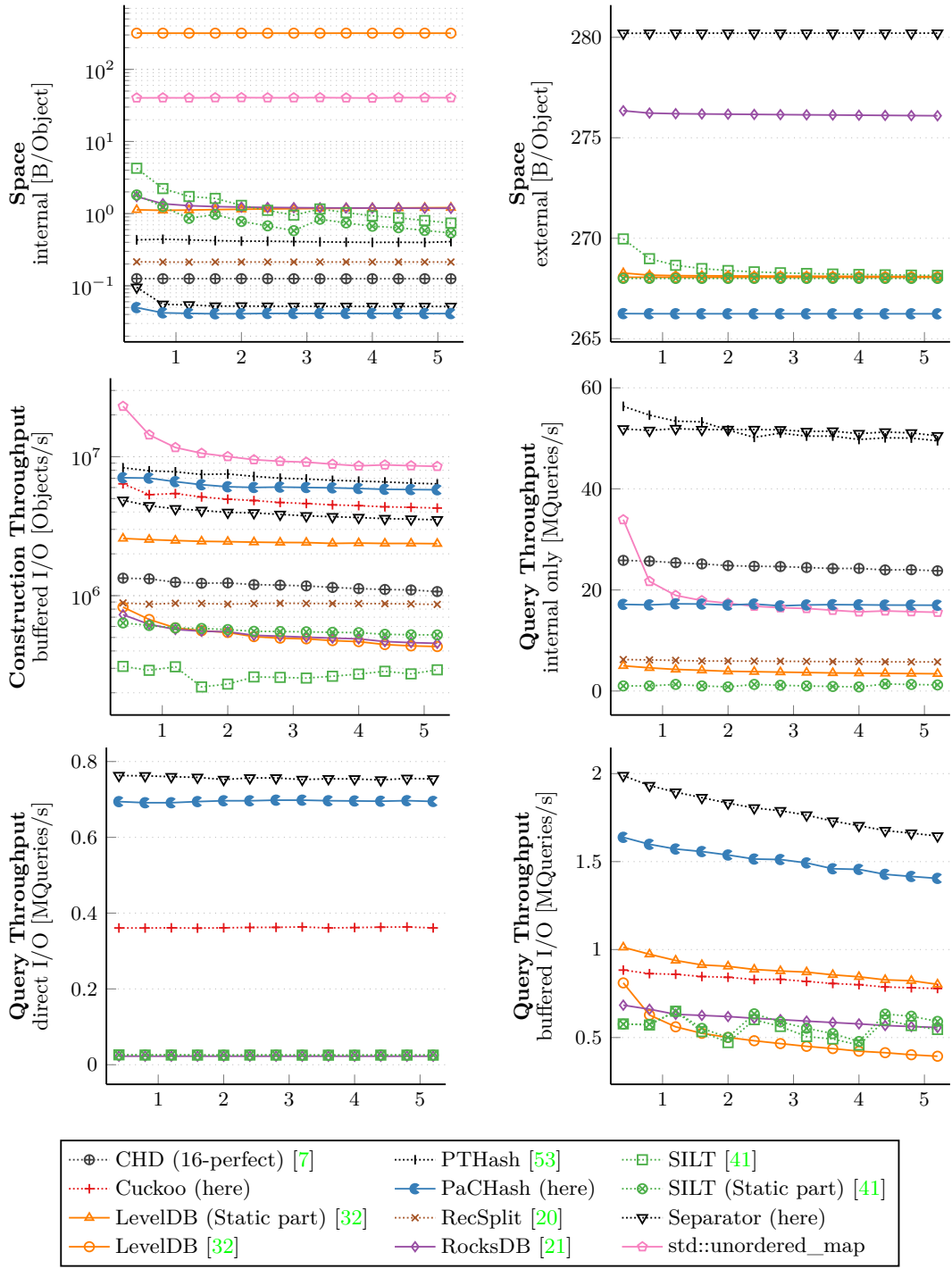}
  \caption{Comparison of object stores using objects of identical size 256 bytes.
  Keys are 8 byte random strings.}
  \label{fig:competitorsIdentical}
\end{figure*}

\begin{figure*}[p]
  \input{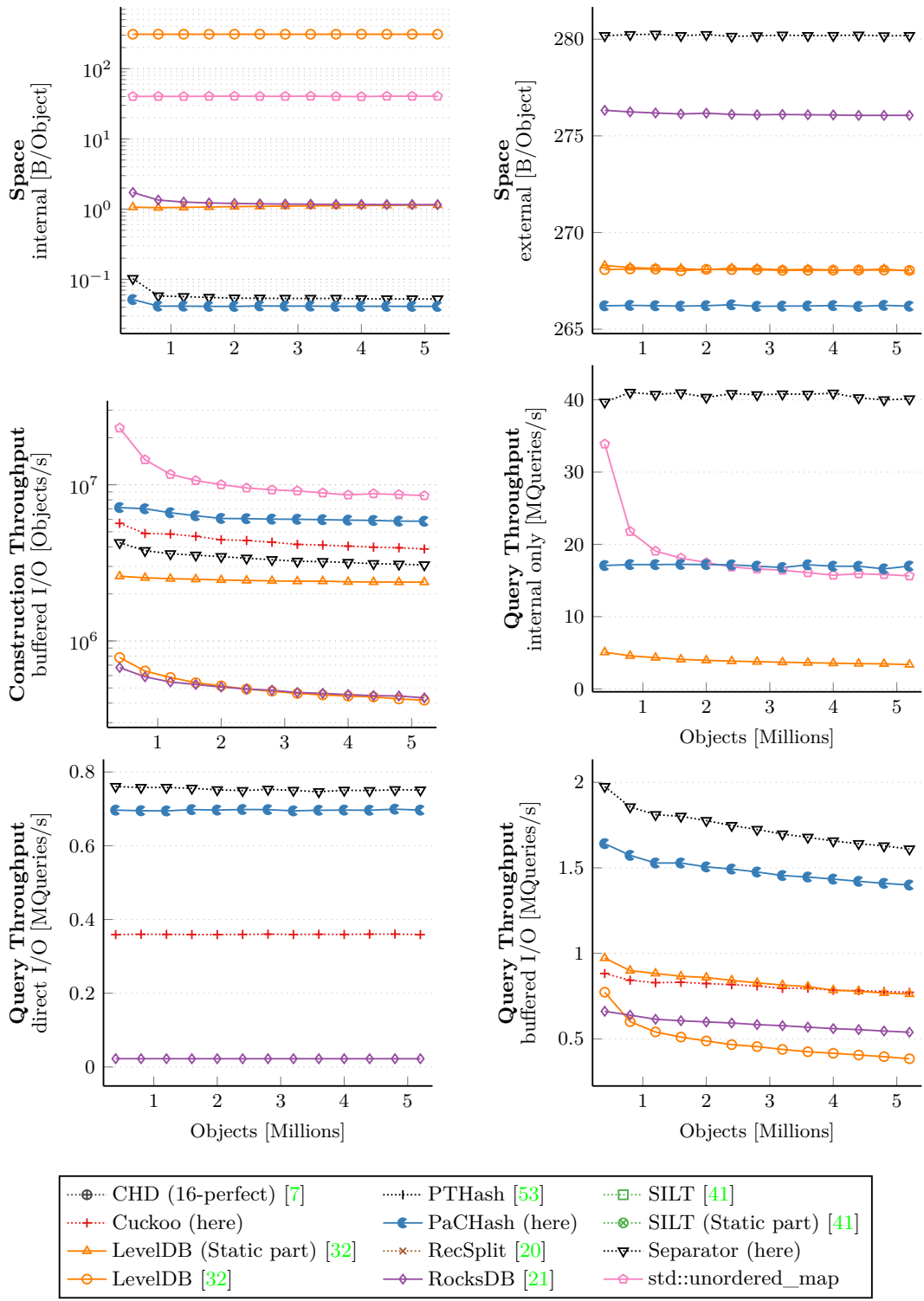}
  \caption{Comparison of object stores using objects of uniform random size $\in [128, 384]$ bytes (bottom).
  Keys are 8 byte random strings.
  Dotted lines indicate methods supporting only objects of identical size natively.
  We enhanced two of them to partially support variable size objects (see \cref{sec:variableSizeAdaption}).}
  \label{fig:competitorsVariable}
\end{figure*}

\begin{table*}[!h]
    \caption{Configurations of competitors}
    \label{tab:competitorConfigurations}
    \centering
      \begin{tabularx}{\textwidth}{l X}
        \toprule
        Competitor                           & Configuration parameters \\ \midrule
        CHD \cite{BelazzouguiBD09}           & Load factor $0.98$. $k=16$ collisions. Bin size $12$.\\
        Cuckoo (here, based on \cite{azar1994balanced,pagh2003basic}) & $2$ alternative positions for each object, loaded in parallel to reduce latency. Streamed queries with \emph{await any}. Load factor $0.95$. Random walk insertion. \\
        LevelDB \cite{google2021leveldb}     & No compression. Construction using a single, large write batch. No Bloom filters. \\
        PaCHash (here)                       & $a=8$. External blocks store a table of keys and offsets. Streamed queries with \emph{await any}. \\
        PTHash \cite{pibiri2021pthash}       & ``Optimizing the general trade-off'' \cite{pibiri2021pthash} with $\alpha=0.94, c=7$, D-D~Encoding. \\
        RecSplit \cite{esposito2020recsplit} & Leaf size $\ell=8$. Bucket size $b=2000$. \\
        RocksDB \cite{facebook2021rocksdb}   & Block cache disabled. No memory mapping or WAL. Queries use batches of size 64. No Bloom filters. \\
        Separator (here, based on \cite{gonnet1988external, larson1984file}) & $6$ bit separators. Load factor $0.96$. Streamed queries with \emph{await any}. \\
        SILT  \cite{lim2011silt}             & \texttt{testCombi.xml} configuration from original repository. \\
        std::unordered\_map                  & $8$ byte keys. $64$ bit pointers to object contents. \\
        \bottomrule
      \end{tabularx}
    \end{table*}

\begin{figure*}[!h]
  \input{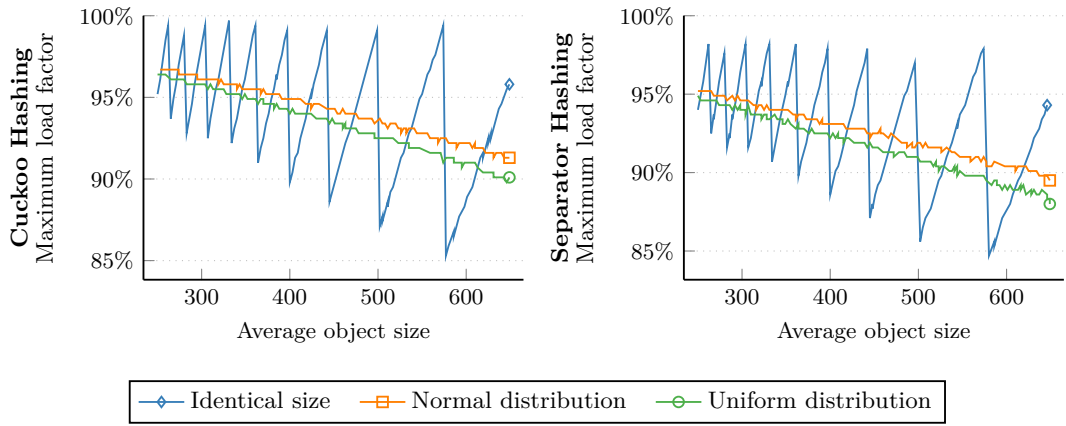}
  \caption{Maximum achievable load factor with different distributions of object sizes of our implementations of Separator Hashing and Cuckoo Hashing that support variable size objects. For an average object size $s$, the normal distribution has a variance of $s/5$ and the uniform random sizes are drawn from $[0.25s, 1.75s]$}
  \label{fig:separatorLoadFactor}
\end{figure*}

Future work might include integrating PaCHash into dynamic external memory object stores, as well as engineering fast and space efficient internal memory variants. On the theoretical side, we would like to better understand the space requirements and lower bounds of bit vectors with entropy coding. This includes relations to different variants of perfect hashing. Although our current analysis assumes random hash functions, PaCHash may also be provably efficient for more realistic simple hash functions.
Further possible space-saving can use the quotienting idea  \cite{koppl2022fast,bender2021all,arbitman2010backyard,Cle84} where some bits of the stored keys are derived from the (now invertible) hash function value. It is interesting how this works best in the presence of nonuniformly distributed keys.

\myparagraph{Acknowledgements.}
The authors would like to thank Peter Dillinger and Stefan Walzer for early discussions leading to this paper.
This project has received funding from the European Research Council (ERC) under the European Union’s Horizon 2020 research and innovation programme (grant agreement No. 882500).

\begin{center}
  \includegraphics[width=4cm]{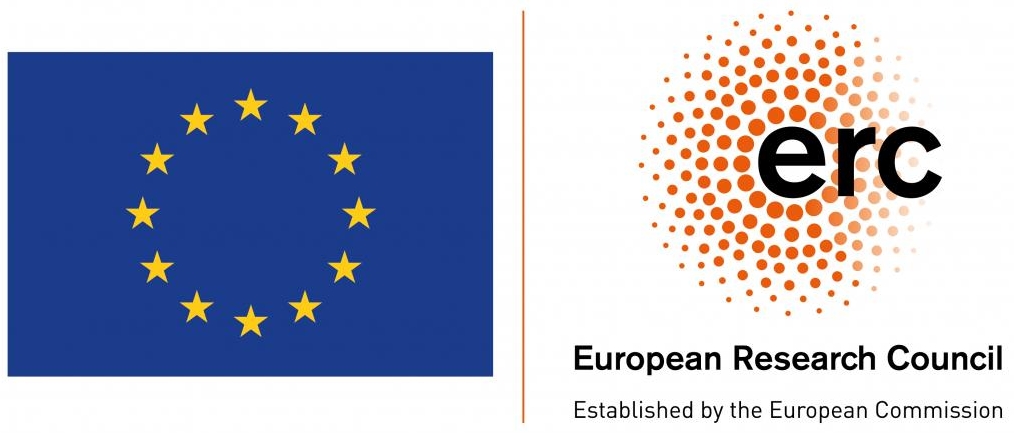}
\end{center}

\clearpage
\bibliographystyle{plainurl}
\bibliography{paper_uniform}

\clearpage

\appendix

\section{Space Usage of Succincter}
Now, we show in more detail how we can achieve the memory requirements of the internal memory index of PaCHash using the Succincter rank and select data structure \cite{Patrascu2008Succincter}.
\begin{proof}(Full Proof of \cref{the:internal_space_requirements_succincter})\label{pro:full_proof_internal_space_requirements}
  Remember that the internal memory data structure \(p\) of PaCHash stores \(m\) integers in the range \(1..am\) and must support predecessor queries.
  We represent all integers in a bit vector of length \((a+1)m\), using the same idea used for the most significant bits in Elias-Fano coding.
  That is, each of the \(m\) integers \(p_i\) is represented as a 1-bit in position \(i+p_i\).
  Answering predecessor queries (which we do not consider here) becomes harder to analyze, as we have no information about the distribution of 1-bits in the bit vector.

  Using Succincter, we can store a size-\(u\) bit vector that contains \(n\) ones and supports rank and select queries using only \(\log\binomial{u}{n}+\frac{u}{\log u}+\Ohtilde{u^{\frac{3}{4}}}\) bits.
  Since we have a length-\((a+1)m\) bit vector that contains \(m\) ones, we require \(\log\binomial{(a+1)m}{m}+\frac{(a+1)m}{\log \left((a+1)m\right)}+\Ohtilde{((a+1)m)^{\frac{3}{4}}}\) bits of space.
  We now show the upper bound for required memory using \cref{lem:bounds_binomial} and \(\Ohtilde{((a+1)m)^{\frac{3}{4}}}=o(m)\).

  \begin{align*}
    &\log\binomial{(a+1)m}{m}+o(m)
     <\log\left(\sqrt{\frac{(a+1)}{2\pi am}}\left(\frac{(a+1)^{a+1}}{a^{a}}\right)^m e^{-\frac{1}{12m+1}}\right)+o(m)\\
     &=\underbrace{\log\sqrt{\frac{(a+1)}{2\pi am}}}_{\leq 0}
       +\log\left(\left(\frac{(a+1)^{a+1}}{a^{a}}\right)^m\right)
       +\underbrace{\log e^{-\frac{1}{12m+1}}}_{\leq 0}+o(m)\\
     &\leq\log\left(\left(\frac{(a+1)^{a+1}}{a^{a}}\right)^m\right)+o(m)
     =m\left((a+1)\log(a+1)-a\log a\right)+o(m)\\
     &=m\left(a\log\left(\frac{a+1}{a}\right)+\log(a+1)\right)+o(m)
     \leq m\left(1.4427+\log(a+1)\right)+o(m)
  \end{align*}

  The last inequality is due to the fact that \(a\log\left(\frac{a+1}{a}\right)\) converges to \(1.4427\approx \frac{1}{\ln 2}\) from below.
  Overall, we require less than \(1.4427+\log(a+1)+o(1)\) bits for each external memory block.
\end{proof}

\begin{lemma}
  Using Succincter for representing monotonic sequences is almost space optimal.
\end{lemma}
\begin{proof}
  In \cref{the:internal_space_requirements_succincter} we have already seen that Succincter needs close to $m(\log(e)+\log(a+1))$ bits of space.
  $\binomial{am}{m}$ is the number of \emph{strictly} monotonic sequences of $m$ numbers in the range $1..am$ and thus a lower bound for the number of monotonic sequences. Using \cref{lem:bounds_binomial} once more, we get

\begin{align*}
  \log\binomial{am}{m}\approx m((a-1)\log\left(\frac{a}{a-1}+\log a\right)
\end{align*}

bits as a lower bound.
Looking at the difference divided by $m$ (i.e. bits per block), we get

\begin{align*}
  a&\log\frac{a+1}{a}+\log(a+1)-(a-1)\log\frac{a}{a-1}-\log a \\
    =&a\log\frac{a^2-1}{a^2}+\log\frac{a+1}{a-1}
    =\frac{\log e}{a}+\Oh{\frac{1}{a^3}}\punkt
\end{align*}

  This difference (obtained using Taylor series development) is much smaller than the $\log e+\log(a+1)$ bits per block needed by the Succincter data structure -- at least for sufficiently large $a$.
\end{proof}

\begin{lemma}
  \label{lem:bounds_binomial}
  For any $c>1, n>0$, let $\displaystyle f(n,c)\Is \sqrt{\frac{c}{(c-1)2\pi n}}\left(\frac{c^c}{(c-1)^{c-1}}\right)^n$, then
  \begin{align*}
    f(c,n)&\left(1-\frac{c^2-c+1}{12c(c-1)n}\right)
    < \binom{cn}{n}
    < f(c,n)e^{-\frac{1}{12n+1}}
    =f(c,n)\left(1-\frac{1}{12n}+\Oh{\frac{1}{n^2}}\right)\punkt
  \end{align*}
\end{lemma}
\begin{proof}
  We use the identity $\binom{cn}{n}=\frac{(cn)!}{n!(cn-n)!}$ as well as Stirling's approximation
  \begin{align*}
    \sqrt{2\pi m}\left(\frac{m}{e}\right)^me^{\frac{1}{12m+1}}<m!<
    \sqrt{2\pi m}\left(\frac{m}{e}\right)^me^{\frac{1}{12m}}\punkt
  \end{align*}
For the upper bound we get
  \begin{align*}
    \binom{cn}{n} &= \frac{(cn)!}{(cn-n)!} \cdot \frac{1}{(cn-n)!} \\
    < & \frac{\sqrt{2\pi cn}\left(\frac{cn}{e}\right)^{cn}e^{\frac{1}{12cn}}}{\sqrt{2\pi n}\left(\frac{n}{e}\right)^{n}e^{\frac{1}{12n+1}}}
        \cdot\frac{1}{\sqrt{2\pi (c-1)n}\left(\frac{(c-1)n}{e}\right)^{(c-1)n}e^{\frac{1}{12(c-1)n+1}}}\\
    =&\sqrt{\frac{c}{(c-1)2\pi n}}\cdot\left(\frac{c^c}{(c-1)^{c-1}}\right)^n
         \cdot e^{{\color{blue}\frac{1}{12cn}}-\frac{1}{12n+1}-{\color{blue}\frac{1}{\smash{\underbrace{\scriptscriptstyle12(c-1)n+1}_{\leq 12cn}}}}}\punkt
  \end{align*}
  The claim follows by observing that the {\color{blue}leftmost} and {\color{blue}rightmost} term in the exponent of $e$ cancel out in the estimation.
  The asymptotic expansion of the upper bound can be obtained using Taylor series expansion.
  
  Similarly, for the lower bound we get
  \begin{align*}
    \binom{cn}{n} &= \frac{(cn)!}{n!}\cdot\frac{1}{(cn-n)!} \\
    > & \frac{\sqrt{2\pi cn}\left(\frac{cn}{e}\right)^{cn}e^{\frac{1}{12cn+1}}}{\sqrt{2\pi n}\left(\frac{n}{e}\right)^{n}e^{\frac{1}{12n}}}
        \cdot\frac{1}{\sqrt{2\pi (c-1)n}\left(\frac{(c-1)n}{e}\right)^{(c-1)n}e^{\frac{1}{12(c-1)n}}}\\
    =&\sqrt{\frac{c}{(c-1)2\pi n}}\cdot\left(\frac{c^c}{(c-1)^{c-1}}\right)^n
        \cdot e^{\frac{1}{12cn+1}-\frac{1}{12n}-\frac{1}{12(c-1)n}}\\
    >&\sqrt{\frac{c}{(c-1)2\pi n}}\cdot\left(\frac{c^c}{(c-1)^{c-1}}\right)^n
        \cdot\left(1-\frac{c^2-c+1}{12c(c-1)n}\right)\punkt
  \end{align*}%
\end{proof}

\end{document}